\documentclass[article, cleveref]{lipics-v2021}
\nolinenumbers

\usepackage{xspace}
\makeatletter
\newcommand{\specificthanks}[1]{\@fnsymbol{#1}}
\makeatother

\usepackage{paralist}
\usepackage{bbold}

\usepackage{algorithm}
\usepackage[noend]{algpseudocode}
\usepackage{cleveref}

\newif\ifdraft

\newcommand{\cardin}[1]{\left| {#1} \right|}%
\newcommand{\etal}{\textit{et~al.} }
\newcommand{\normX}[1]{\left\| #1 \right\|}

\newcommand{\lemlab}[1]{\label{lemma:#1}}
\newcommand{\lemref}[1]{Lemma \ref{lemma:#1}}%

\newcommand{\thmlab}[1]{{\label{theo:#1}}}
\newcommand{\thmref}[1]{Theorem \ref{theo:#1}}

\newcommand{\apdxlab}[1]{{\label{apdx:#1}}}
\newcommand{\apdxref}[1]{Appendix \ref{apdx:#1}}

\usepackage{pgfplots}
\usepackage{amsmath,amsfonts,amssymb,mathtools}
\usepackage{graphicx,float}
\usepackage[short]{optidef}

\DeclareMathOperator*{\argmin}{arg\,min}

\newcommand{\dsg}{\textsc{DSG}\xspace}
\newcommand{\dss}{\textsc{DSS}\xspace}
\newcommand{\greedy}{\textsc{Greedy}\xspace}
\newcommand{\sgreedy}{\textsc{SuperGreedy}\xspace}
\newcommand{\gplus}{\textsc{Greedy++}\xspace}
\newcommand{\sgplus}{\textsc{SuperGreedy++}\xspace}

\definecolor{nalmostblack}{rgb}{0, 0, 0.7}

\usepackage{amsmath,amssymb}

\newcommand{\Ex}[1]{\mathop{\mathbb{E}}\left[ #1 \right]}

\newcommand{\mypara}[1]{\smallskip \noindent {\bf #1}}
\newcommand{\eps}{\epsilon}

\author{Elfarouk {Harb}}{University of Illinois at Urbana-Champaign}{eyfmharb@gmail.com}{}{Supported in part by NSF grant CCF-2028861}
\author{Kent {Quanrud}}{Purdue University}{krq@purdue.edu}{}{Supported in part by NSF grant CCF-2129816}
\author{Chandra {Chekuri}}{University of Illinois at Urbana-Champaign}{chekuri@illinois.edu}{}{Supported in part by NSF grants CCF-2028861 and CCF-1910149.}

\authorrunning{Elfarouk Harb, Kent Quanrud, and Chandra Chekuri}
\Copyright{Licensed under Creative Commons License CC-BY 4.0}

\title{Convergence to Lexicographically Optimal Base in a (Contra)Polymatroid and
  Applications to Densest Subgraph and Tree Packing}

\titlerunning{Convergence to Lexicographically Optimal Base and Applications}

\begin{document}

\pagenumbering{gobble}

\maketitle
\begin{abstract}
  Boob et al. \cite{flowless} described an iterative peeling algorithm
  called \gplus for the Densest Subgraph Problem (\dsg) and
  conjectured that it converges to an optimum solution.  Chekuri,
  Qaunrud and Torres \cite{chandra-soda} extended the algorithm to
  general supermodular density problems (of which DSG is a special
  case) and proved that the resulting algorithm
  \textsc{Super-Greedy++} (and hence also \textsc{Greedy++})
  converges. In this paper we revisit the convergence proof and
  provide a different perspective.  This is done via a connection to
  Fujishige's quadratic program for finding a lexicographically optimal base
  in a (contra) polymatroid \cite{fujishige}, and a noisy version of the Frank-Wolfe
  method from convex optimization \cite{FW-56,pmlr-v28-jaggi13}.  This gives us a simpler convergence
  proof, and also shows a stronger property that \textsc{Super-Greedy++}
  converges to the optimal dense decomposition vector, answering a
  question raised in Harb et al. \cite{farouk-neurips}. A second
  contribution of the paper is to understand Thorup's work on ideal tree packing and 
  greedy tree packing \cite{Thorup07,Thorup08} via
  the Frank-Wolfe algorithm applied to find a lexicographically optimum
  base in the graphic matroid.  This yields a simpler and transparent proof. The two results appear disparate but are unified via Fujishige's result and convex optimization. 
\end{abstract}

\newpage

\pagenumbering{arabic}

\section{Introduction}
In this paper we consider iterative greedy algorithms for two
different combinatorial optimization problems and show that the
convergence of these algorithms can be understood by combining two
general tools, one coming from the theory of submodular functions, and
the other coming from convex optimization. This yields simpler proofs via
a unified perspective, while also yielding additional properties that
were previously unknown.

\mypara{Densest subgraph and supermodularity:} We start with the
initial problem that motivated this work, namely, the densest subgraph
problem (\dsg). The input to \dsg is an undirected graph $G=(V,E)$
with $m=\cardin{E}$ and $n=\cardin{V}$. The goal is to return a subset
$S\subseteq V$ that maximizes $\frac{\cardin{E(S)}}{\cardin{S}}$ where
$E(S)=\{uv \in E : u,v\in S \}$ is the set of edges with both end
points in $S$. Throughout the paper, we let
$\lambda(G)=\frac{\cardin{E(G)}}{\cardin{V(G)}}$ denote the density of
graph $G(V,E)$. We treat the unweighted case for simplicity; all the
results generalize to edge-weighted graphs. Goldberg \cite{goldberg}
and Picard and Queyranne \cite{pq-82} showed that \dsg can be
efficiently solved via a reduction to the $s$-$t$ maximum-flow
problem. 

A different connection that shows polynomial-time solvability
of \dsg is important to this paper. Consider a real-valued set
function $f:2^V \rightarrow \mathbb{R}_+$ defined over the vertex set
$V$, where $f(S) = |E(S)|$. This function is \emph{supermodular}.  A
function $f$ is supermodular iff $-f$ is \emph{submodular}. A
real-valued set function $f:2^V \rightarrow \mathbb{R}$ is submodular
iff $f(A) + f(B) \ge f(A \cup B) + f(A \cap B)$ for all $A, B
\subseteq B$. Submodular and supermodular set functions are
fundamental in combinatorial optimization --- see
\cite{Schrijver-book,Fujishige-book}.  

Coming back to \dsg, maximizing
$|E(S)|/|S|$ is equivalent to finding the largest $\lambda$ such that
$\lambda |S| - |E(S)| \ge 0$ for all $S \subseteq V$.  This
corresponds to minimizing the submodular function $g$ where $g(S) =
\lambda |S| - |E(S)|$.  A classical result in combinatorial
optimization is that the minimum of a submodular set function can be
found in polynomial-time in the value oracle setting. Thus, \dsg can
be solved via reduction to submodular set function minimization and binary search.
The preceding connection also motivates the definition of a generalization of \dsg
called the densest supermodular set problem (\dss) (see
\cite{chandra-soda}). The input is a non-negative supermodular
function $f:2^V\rightarrow \Re_+$, and the goal is to find $S\subseteq
V$ that maximizes $\frac{f(S)}{|S|}$.  \dss is polynomial-time
solvable via submodular set function minimization.  \dsg, \dss and its variants have several applications in practice, and they are routinely used in graph and network analysis to
find dense clusters or communities. We refer the reader to the
extensive literature on this topic \cite{dsg-chenhao, flowless, frankwolfe, CharalamposWWW, bintao-kclist, reid-dsg, Charalampos-sigkdd, AlessandroWWW, andrewMc-math-foundations, Polina-PKDD, Balalau-WWW, Kuroki2020OnlineDS, dense-maentenance, Kijung-k-core, Li2020FlowScopeSM, Lanciano2023survey}. \dsg is also of 
interest in algorithms via its connection to arboricity and related notions --- see \cite{SawlaniWang,Christiansenetal} for recent work.

\mypara{Faster algorithms, Greedy and Greedy++:} Although \dsg is
polynomial-time solvable via maxflow and submodular function
minimization, the corresponding algorithms are not
yet practical for the large graphs that arise in many
applications; this is despite the fact that we now have very fast theoretical algorithms for
maxflow and mincost flow \cite{veryfast-max-flow}. 
For this reason there has been considerable interest in
faster (approximation) algorithms. More than 20 years ago Charikar \cite{charikar} showed
that a simple ``peeling'' algorithm (\greedy) yields a
$1/2$-approximation for \dsg. An ordering of the vertices as
$v_{i_1},v_{i_2},\ldots,v_{i_n}$ is computed as follows: $v_{i_1}$ is
a vertex of minimum degree in $G$ (ties broken arbitrarily), $v_{i_2}$
is a minimum degree vertex in $G-v_{i_1}$ and so on\footnote{This
peeling order is the same as the one used to create the so-called core
decomposition of a graph \cite{core-survey} and the \greedy algorithm itself was
suggested by Asahiro et al. \cite{aitt-00}.}. After creating the ordering,
the algorithm picks the best suffix, in terms of density, among the
$n$-possible suffixes of the ordering. Charikar also developed a
simple exact LP relaxation for \dsg. Charikar's results have been
quite influential.  \greedy can be implemented in (near)-linear time
and has also been adapted to other variants. The LP relaxation has
also been used in several algorithms that yield a
$(1-\eps)$-approximate solution \cite{bgm-14,bsw-19}, and has led to a
flow-based $(1-\eps)$-approximation \cite{chandra-soda}.  More
recently, Boob et al. \cite{flowless} developed an algorithm called
\gplus that is based on combining \greedy with ideas from
multiplicative weight updates (MWU); the algorithm repeatedly applies a
simple peeling algorithm with the first iteration coinciding with
\greedy but later iterations depending on a weight vector that is
maintained on the vertices --- the formal algorithm is described in a
later section. The advantage of the algorithm is its simplicity, and
Boob et al.\ \cite{flowless} showed that it has very good empirical
performance. Moreover they conjectured that \textsc{Greedy++}
converges to a $(1-\eps)$-approximation in $O(1/\eps^2)$
iterations. Although their strong conjecture is yet unverified,
Chekuri et al.\ \cite{chandra-soda} proved that \gplus converges in
$O(\frac{\Delta \log |V|}{\eps^2 \lambda(G)})$ iterations where
$\Delta$ is the maximum degree of $G$.

The convergence proof in \cite{chandra-soda} is non-trivial. The proof relies crucially in considering \dss and supermodularity. \cite{chandra-soda} shows that \greedy and \gplus can be
generalized to \sgreedy and \sgplus for \dss, and that \sgplus
converges to a $(1-\eps)$-approximation solution in
$O(\alpha_f/\eps^2)$ iterations where $\alpha_f$ depends (only) on the
function $f$.

\mypara{Dense subgraph decomposition and connections:} As we
discussed, \dsg is a special case of \dss and hence \dsg inherits
certain nice structural properties from supermodularity. One of these
is the fact that the vertex set $V$ of every graph $G=(V,E)$ admits a
decomposition into $S_1,S_2,\ldots,S_k$ for some $k$ where
$S_1$ is the vertex set of the \emph{unique maximal} densest subgraph,
$S_2$ is the unique maximal densest subgraph after ``contracting''
$S_1$, and so on. This fact is easier to see in the setting of \dss.
Here, the fact that $S_1$ is the \emph{unique} maximal densest set and this follows from
supermodularity; if $A$ and $B$ are optimum dense sets then so is $A
\cup B$.  One can then consider a new supermodular function
$f_{S_1}:2^{V-S_1} \rightarrow \mathbb{R}$ defined over $V-S_1$ where
$f_{S_1}(A) = f(S_1 \cup A) - f(S_1)$ for all $A \subseteq V-S_1$.
The new function is also supermodular. Then $S_2$ is the unique
maximal densest set for $f_{S_1}$. We iterate this process until we obtain an empty
set. The decomposition also allows us to assign a density value
$\lambda_v$ to each $v \in V$ (which corresponds to the density of the
set when $v$ is in the maximal set). We call this the density vector
associated with $f$. Dense
decompositions follow from the theory of principal partitions of
submodular functions \cite{Narayanan91,Narayanan-book,Fujishige-survey}.
In the context of graphs and DSG this was rediscovered by Tatti and Gionis who called it the locally-dense decomposition \cite{tg-15,t-19}, and gave algorithms for computing it. Subsequently, Danisch \etal \cite{frankwolfe} applied
the well-known Frank-Wolfe algorithm for constrained convex
optimization to a quadratic program derived from Charikar's LP
relaxation for \dsg.  More recently, Harb et
al.\ \cite{farouk-neurips} obtained faster algorithms for computing
the dense decomposition in graphs via Charikar's LP; they used
a different method called FISTA for constrained convex optimization based on
acceleration. Although \dss was not the main focus,
\cite{farouk-neurips} also made an important connection to 
Fujishige's result on lexicographically optimal base in polymatroids \cite{fujishige} which elucidated the work of Danisch et al.\ on \dsg. We describe this next.

\mypara{Lexicographical optimal base and dense decomposition:} We
briefly describe Fujishige's result \cite{fujishige} and its
connection to dense decompositions. Let $f:2^V \rightarrow
\mathbb{R}_+$ be a monotone submodular set function ($f(A) \le f(B)$
if $A \subset B$) that is also normalized ($f(\emptyset) = 0$).
Following Edmonds, the polymatroid associated with $f$, denote by
$P_f$ is the polyhedron
\begin{math}
    \{ x \in \mathbb{R}^V \mid x \ge 0, x(S) \le f(S) \quad \forall S \subseteq V\},
\end{math}
where $x(S)=\sum_{i\in S}x_i$.  The base polyhedron associated
with $f$, denote by $B_f = P_f \cap \{ x \in \mathbb{R}^V \mid x(V) =
f(V)\}$ obtained by intersecting $P_f$ with the equality constraint
$x(V) = f(V)$. Each vector $x$ in $B_f$ is called a base.  If $f$ is a
monotone normalized supermodular function we consider the
contrapolymatroid $P_f = \{ x \in \mathbb{R}^V
\mid x \ge 0, x(S) \ge f(S) \quad \forall S \subseteq V\}$ (the
inequalities are reversed), and similarly $B_f$ is the base
contrapolymatroid obtained by intersecting $P_f$ with equality
constraint $x(V) = f(V)$.  Fujishige proved that there exists a unique
lexicographically minimal base in any polymatroid, and morover it can
found by solving the quadratic program: $\min \sum_v x_v^2
\mbox{~s.t~} x \in B_f$.  In the context of supermodular functions, one obtains a
similar result; the quadratic program $\min \sum_v x_v^2 \mbox{~s.t~}
x \in B_f$ where $B_f$ is contrapolymatroid associated with $f$ has a
unique solution. As observed explicity in
\cite{farouk-neurips}, the lexicographically optimal base gives the
dense decomposition vector for \dss.  That is, if $x^*$ is the optimal
solution to the quadratic program then for each $v$, $x^*_v =
\lambda_v$. In particular, as noted in \cite{farouk-neurips}, one can
apply the well-known Frank-Wolfe algorithm to the quadratic program
and it converges to the dense decomposition vector. As we will
see later, each iteration corresponds to finding a maximum weight base
in a contrapolymatroid which is easy to find via the greedy
algorithm.

\mypara{(Ideal) Tree packings in graphs and the Tutte--Nash-Williams theorem:} Our discussion
so far focused on \dsg. Now we describe a different
problem on graphs and relevant background. As we said, our goal is to present a unified perspective on these two problems. The well-known Tutte--Nash-Williams theorem in graph theory (see \cite{Schrijver-book}) establishes a min-max
result for the maximum number of edge-disjoint spanning trees in a
multi-graph $G$. Given an undirected graph $G= (V,E)$, and a partition
$P$ of the vertices, let $E(P)$ denote the number of edges crossing
from one partition to another. We say the strength of a partition is
$\frac{E(P)}{|P|-1}$. Let $\mathcal{T}(G)$ denote all possible
spanning trees of $G$. Let $\tau^*(G)$ denote the maximum number of
edge-disjoint spanning trees in $G$. Then $\tau^*(G) = \min_P \lfloor
\frac{E(P)}{|P|-1} \rfloor$. Further, if
we define $\tau(G)$ to be the maximum \emph{fractional} packing of
spanning trees, then the floor can be removed and we have $\tau(G) =
\min_P \frac{E(P)}{|P|-1}$. We note that the graph theoretic result is a special case of matroid base packing. Tree packings are useful for a number of
applications. In particular, Karger \cite{Karger00} used tree packings and
other ideas in his well-known near-linear randomized algorithm for
computing the global minimum cut of a graph. We are mainly concerned here with
Thorup's work in \cite{Thorup07,Thorup08} that was motivated by dynamic mincut and $k$-cut problems. He defined the so-called \emph{ideal} edge loads and ideal tree packing (details in later section) by recursively decomposing the graph via Tutte--Nash-Williams partitions \cite{Thorup07}. 
He also proved that a simple iterative greedy tree packing algorithm converges to the ideal
loads \cite{Thorup08}. He used the approximate ideal tree packing to obtain
new deterministic algorithms for the $k$-cut problem, and his approach
has been quite influential in a number of subsequent results \cite{fukunaga2010computing,ChekuriQX20,lokshtanov2022parameterized,li2019faster,Gupta2021optimal}. 
Thorup obtained his tree packing result from first principles. We ask: is there a connection between ideal tree packing and \dsg?

\subsection{Contributions of the paper}
This paper has two main contributions. The first is a new proof
of the convergence of \sgplus for \dss.  Our proof is
based on showing that \sgplus can be viewed as a ``noisy'' or
``approximate'' variant of the Frank-Wolfe algorithm applied to the
quadratic program defined by Fujishige.  The advantage of the new
proof is twofold. First, it shows that \sgplus not only converges to a
$(1-\eps)$-approximation to the densest set, but that in fact it
converges to the densest decomposition vector.  This was empirically
observed in \cite{farouk-neurips} for \dsg, and was left as an open
problem to resolve. The proof in \cite{chandra-soda} on convergence of \sgplus 
is based on the MWU method via LPs, and does not exploit Fujishige's result which is key to the stronger property that we prove here. Second,
the proof connects two powerful tools directly and at a high-level:
Fujishige's result on submodular functions, and a standard method for constrained convex optimization.

\begin{theorem}
  \label{thm:intro-greedyplusplus}
  Let $b^*$ be the dense decomposition vector for a non-negative
  monotone supermodular set function $f:2^V \rightarrow \mathbb{R}_+$
  where $|V| = n$. Then, $\sgplus$ converges in $O(\alpha_f/\eps^2)$
  iterations to a vector $b$ such that $||b-b^*||_2 \le \eps$, where
  $\alpha_f$ depends only on $f$. For a graph with $m$ edges and $n$
  vertices, $\gplus$ converges in $O(mn^2/\eps^2)$
  iterations for unweighted multigraphs.
\end{theorem}

\begin{remark}
  The new convergence gives a weaker bound than the one in \cite{chandra-soda}
  in terms of convergence to a $(1-\eps)$ \emph{relative} approximation to the maximum density. However, it gives a strong \emph{additive} guarantee to the \emph{entire} dense decomposition vector.
\end{remark}

Our second contribution builds on our insights on \dsg
and \dss, and applies it towards understanding ideal tree packing and 
greed tree packing. We connect the ideal tree packing of Thorup to the dense
decomposition associated with the rank function of the underlying
graphic matroid (which is submodular). We then show that greedy
tree packing algorithm can be viewed
as the Frank-Wolfe algorithm applied to the quadratic
program defined by Fujishige, and this easily yields a convergence guarantee.

\begin{theorem}
  \label{thm:intro-thorup}
  Let $G=(V,E)$ be a graph. The ideal edge load vector $\ell^*: E \rightarrow \mathbb{R}_+$ for $G$ is given by the lexicographically minimal base in the polymatroid associated with the rank function of the graphic matroid of $G$. The Frank-Wolfe algorithm with step size $\frac{1}{k+1}$, when applied to the quadratic program for computing the lexicographically minimal base in the graphic matroid of $G$, coincides with the greedy tree packing algorithm. For unweighted graphs on $m$ edges, the generic analysis of Frank-Wolfe method's convergence shows that greedy tree packing converges to a load vector $\ell:E \rightarrow \mathbb{R}_+$ such that $||\ell-\ell^*|| \le \eps$ in $O(\frac{m \log (m/\eps)}{\eps^2})$ iterations. The standard step size algorithm converges in $O(\frac{m}{\eps^2})$ iterations.
\end{theorem}

\begin{remark}
  Although the algorithm is the same (greedy tree packing), Thorup's analysis guarantees a strongly polynomial-bound even in the capacitated case \cite{Thorup08}. However we obtain a stronger additive guarantee via a \emph{generic} Frank-Wolfe analysis and our analysis has a $1/\epsilon^2$ dependence while Thorup's has a $1/\epsilon^3$ dependence. We give a more detailed comparison in Section~\ref{sec:treepacking}.
\end{remark}

\mypara{Organization:} The rest of the paper is devoted to proving the two theorems.
The paper relies on tools from theory of submodular functions and an adaptation
of the analysis of Frank-Wolfe method. We first describe the relevant  background
and then prove the two results in separate sections. Due to space constraints,  most of the proofs are provided in the appendix. A future version will discuss additional related work in more detail.

\section{Background on Frank-Wolfe algorithm and a variation}
\label{sec:frankwolfe}
Let $\mathcal{D}\subseteq \Re^d$ be a compact convex set, and
$f:\mathcal{D}\rightarrow \Re$ be a convex, differentiable
function. Consider the problem of $\min_{x \in \mathcal{D}} f(x)$.
Frank-Wolfe method \cite{FW-56} is a first order method and it relies on access to a
linear minimization oracle, \textsc{LMO}, for $f$ that can answer
$\textsc{LMO}(w)=\argmin\limits_{s\in \mathcal{D}} \langle s, \nabla
f(w) \rangle$ for any given $w\in \mathcal{D}$. In several applications
such oracles with fast running times exist. Given $f, \mathcal{D}$ as
above, the Frank-Wolfe algorithm is an iterative algorithm that converges
to the minimizer ${\bf x}^\ast \in \mathcal{D}$ of $f$. See
Algorithm \ref{Frank-Wolfe-Original}. The algorithm
starts with a guess of the minimizer $b^{(0)}\in \mathcal{D}$. In each
iteration, it finds a direction $d^{(k+1)}$ to move towards by calling
the linear minimization oracle on the current guess $b^{(k)}$. It then
moves slightly towards that direction using a convex combination to
ensure that the new point is in $\mathcal{D}$. The amount the algorithm moves
towards the new direction decreases as $k$ increases signifying the
``confidence'' in its current guess as the minimizer.

\begin{algorithm}
  \begin{algorithmic}[1]
    \State Initialize $b^{(0)} \in \mathcal{D}$
    \For {$k\leftarrow 0$ to $T-1$} 
        \State $\gamma \leftarrow \frac{2}{k+2}$
        \State $d^{(k+1)} \leftarrow \argmin\limits_{s \in \mathcal{D}} (\langle s, \nabla f(b^{(k)})  \rangle)$ \Comment{Call oracle on $b^{(k)}$}
        \State $b^{(k+1)}\leftarrow (1-\gamma)b^{(k)}+\gamma d^{(k+1)}$
    \EndFor{}
  \Return $b^{(T)}$
 \end{algorithmic}
  \caption{\textsc{Frank-Wolfe-Original}} \label{Frank-Wolfe-Original}
 \end{algorithm}

The original convergence analysis for the Frank-Wolfe algorithm is from
\cite{FW-56}. Jaggi \cite{pmlr-v28-jaggi13} gave an elegant and simpler
analysis. His analysis  characterizes the convergence rate in terms of the \textit{curvature
  constant} $C_f$ of the function $f$.

\begin{definition}
Let $\mathcal{D}\subseteq \Re^d$ be a compact convex set, and $f:\mathcal{D}\rightarrow \Re$ be a convex, differentiable function. The curvature constant $C_f$ of $f$ is defined as 
$$C_f = \sup_{x, s \in D, \gamma \in [0,1], y=x+\gamma(s-x)} \frac{2}{\gamma^2} (f(y)-f(x) - \langle y-x, \nabla f(x) \rangle ).$$
\end{definition}

\begin{definition}
Let $g:\mathcal{D}\rightarrow \Re$ be a differentiable function. Then $g$ is Lipschitz with constant $L$ if for all $x,y\in \mathcal{D}$, $\normX{g ({\bf x})-g({\bf y})}_2 \leq L \normX{x-y}_2$.
\end{definition}

Let $\text{diam}(\mathcal{D}) = \max\limits_{x,y \in \mathcal{D}} \normX{x-y}_2$ be the diameter of $\mathcal{D}$.
One can show that $C_f \le L \cdot \text{diam}(\mathcal{D})^2$ where $L$ is the Lpischitz constant of $\nabla f$.

\begin{theorem}[\cite{{pmlr-v28-jaggi13}}]
Let $\mathcal{D}\subseteq \Re^d$ be a compact convex set, and $f:\mathcal{D}\rightarrow \Re$ be a convex, differentiable function with minimizer ${\bf b}^\ast$. Let ${\bf b}^{(k)}$ denote the guess on the $k$-th iteration of the Frank-Wolfe algorithm. Then
$f({\bf b}^{(k)}) - f({\bf b}^\ast) \leq \frac{2C_f}{k+2}$.
\end{theorem}

Jaggi's proof technique can be used to prove the
convergence rate of ``noisy/approximate'' variants of the Frank-Wolfe
algorithm. This motivates the following definition.  An
\textit{$\epsilon$-approximate linear minimization oracle} is an
oracle that for any ${\bf w}\in \mathcal{D}$, returns $\hat{{\bf s}}$
such that $\langle {\hat{\bf s}}, \nabla f({\bf w}) \rangle \leq
\langle {\bf s}^\ast, \nabla f({\bf w}) \rangle + \epsilon$, where
$s^\ast = \textsc{LMO}({\bf w})$.
While an efficient \textit{exact} linear minimization oracle 
exists in some applications, in others one can only
$\epsilon$-approximate it (using numerical methods or otherwise). Jaggi's proof technique
extends to show that an approximate linear minimization
oracles suffices for convergence as long as the approximation quality improves
with the iterations. Suppose the oracle, in iteration $k$, provides
a $\frac{\delta C_f}{k+2}$-approximate solution where $\delta>0$ is some fixed constant.
The convergence rate will only deteriorate by a
$(1+\delta)$ multiplicative factor.
Qualitatively, this says that we
can afford to be inaccurate in computing the Frank-Wolfe direction in
early iterations, but the approximation should approach
$\textsc{LMO}(b^{(k)})$ as $k\rightarrow \infty$. 

Another question of interest is the resilience of the Frank-Wolfe
algorithm to changes in the learning rate
$\gamma_{k}=\frac{2}{k+2}$. Indeed, the variants we will look at will
\textit{require} $\gamma_k=\frac{1}{k+1}$. As we will see, Jaggi's
proof can again be adapted to handle this case, with only an $O(\log
k)$ multiplicative deterioration in the convergence rate. We state the
following theorem whose proof we defer to the appendix.

\begin{theorem}{[Proof in \apdxref{pf:FW-Resistant}]}
\thmlab{FW-Resistant} Let $\mathcal{D}\subseteq \Re^d$ be a compact
convex set, and $f:\mathcal{D}\rightarrow \Re$ be a convex,
differentiable function with minimizer ${\bf b}^\ast$. Suppose instead
of computing ${\bf d}^{(k+1)}$ by calling $\textsc{LMO}({\bf
  b^{(k)}})$ in iteration $k$, we call a $\frac{\delta
  C_f}{k+2}$-approximate linear minimization oracle, for some fixed
$\delta>0$. Also, suppose instead of using $\gamma_k=\frac{2}{k+2}$,
we use $\gamma_k = \frac{1}{k+1}$ as a step size. Then
$f({\bf b}^{(k)}) - f({\bf b}^\ast) \leq \frac{2C_f(1+\delta)H_{k+1}}{k+1}$, where $H_n$ is the $n$-th Harmonic term. 
\end{theorem}

We refer to the variant of Frank-Wolfe algorithm as described by \thmref{FW-Resistant} as \emph{noisy} Frank-Wolfe.

\section{Sub and supermodular functions, and dense decompositions}
\label{sec:submod-background}
We already defined submodular and supermodular set functions, polymatroids and contrapolymatroids.
We restrict attention to functions satisfying $f(\emptyset)=0$ which together
with supermodularity and non-negativity implies monotonocity, that is, $f(A)\leq
f(B)$ for $A\subseteq B$. An alternative definition of submodularity
is via diminishing marginal values. We let $f(v \mid A) = f(A \cup \{v\}) - f(A)$ denote the marginal value of $v$ to $A$. Submodularity is equivalent to $f(v \mid A) \ge f(v \mid B)$
whenever $A \subseteq B$ and $v \in V \setminus B$; the inequality is
reversed for supermodular set functions. 
We need the following simple lemma.
\begin{lemma}{[Proof in \apdxref{pf:submodularnegissuper}]}
\lemlab{submodularnegissuper}
For a submodular function $f:2^V \rightarrow \Re$, the function $g(X)=f(V)-f(V\setminus X)$ is supermodular.
In particular if $f$ is a normalized monotone submodular function then $g$ is a normalized monotone supermodular function.
\end{lemma}

\mypara{Deletion and contraction, and non-negative summation:} Sub and
supermodular functions are closed under a few simple operations.
Given $f:2^V \rightarrow \mathbb{R}$, restricting it to a subset $V'$ corresponds to deleting
$V \setminus V'$. Given $A \subset V$, contracting $f$ to $A$ yields
the function $g:2^{V \setminus A} \rightarrow \mathbb{R}$ where $g(X)
= g(X \cup A) - g(A)$. Given two functions
$f$ and $g$ we can take their non-negative sum $a f + b g$ where $a, b
\ge 0$.  Monotonicity and normalization is also preserved under these operations.

\subsection{Dense decompositions for submodular and supermodular functions}
Following the discussion in the introduction, we are interested in
decompositions of supermodular and submodular functions. Dense
decompositions follow from the theory of principal partitions of
submodular functions that have been explored extensively. We refer the
reader to Fujishige's survey \cite{Fujishige-survey} as well as
Naraynan's work \cite{Narayanan91,Narayanan-book}. The standard
perspective comes from considering the minimizers of the function
$f_{\lambda}$ for a scalar $\lambda$ where
$f_{\lambda}(S) - \lambda |S|$. As $\lambda$ varies from $-\infty$ to
$\infty$ the minimizers change only at a finite number of break
points. In this paper we are interested in the notion of density, in
the form of ratios, for non-negative submodular and supermodular
functions. For this reason we follow the notation from recent work
\cite{t-19,frankwolfe,chandra-soda,farouk-neurips} and state lemmas in
a convenient form, and provide proofs in the appendix for the sake of
completeness.

\mypara{Supermodular function dense decomposition:}
The basic observation is the following.
\begin{lemma}{[Proof in \apdxref{pf:sprmod:unqqq}]}
\lemlab{sprmod:unqqq}
    Let $f: 2^V \rightarrow \Re_+$ be a non-negative supermodular set
function. There exists a \emph{unique maximal} set $S\subseteq V$ that maximizes $\frac{f(S)}{\cardin{S}}$.  
\end{lemma}

The preceding lemma can be used in a simple fashion to derive the following corollary
(this was explicitly noted in \cite{chandra-soda} for instance).
\begin{corollary}
\label{unique-decomp}
   Let $f: 2^V \rightarrow \Re_+$ be a non-negative supermodular set
   function. There is a unique partition $S_1,S_2,\ldots,S_h$ of $V$ with the following property.
    Let $V_i = V - \cup_{j<i} S_j$ and let $A_i = \cup_{j<i}
   S_i$. Then,   for each $i = 1$ to $h$, $S_i$ is the unique maximal densest
   set for the function $f_{D_i}:2^{V_i} \rightarrow
   \mathbb{R}_+$. Moroever, letting $\lambda_i$ be the optimum density
   of $f_{D_i}$, we have $\lambda_1 > \lambda_2 \ldots > \lambda_h$.
 \end{corollary}

 Based on the preceding corollary, we can associated with each $v \in 
 V$ a value $\lambda(v)$: $\lambda(v) = \lambda_i$ where $v \in S_i$. See Figure \ref{fig:dense_decomposition_example} for an example of a dense decomposition of the function $f(S)=\cardin{E(S)}$.  
   
\mypara{Dense decomposition for submodular functions:} We now discuss submodular functions.
  We consider two variants. We start with a basic observation.

\begin{lemma}{[Proof in \apdxref{pf:minimality:submodular}]}
  \lemlab{minimality:submodular}
    Let $f:2^V \rightarrow \Re_+$ be a monotone non-negative submodular set function
    such that $f(v) > 0$ for all $v \in V$.
    There is a unique minimal set $S\subseteq V$ that minimizes
    $\frac{\cardin{V}-\cardin{S}}{f(V)-f(S)}$ for submodular function
    $f$.
\end{lemma}

Consider the following variant of a decomposition of $f$. We let $S_0
= V$ and find $S_1$ as the unique \textit{minimal} set $S\subseteq V$
that minimizes $\frac{|V|-|S|}{f(V)-f(S)}$. Then we ``delete''
$\hat{S_1} = V\setminus S_1$, and find the minimal set $S_2\subseteq
S_1 $ that minimizes $\frac{|S_1|-|S|}{f(S_1)-f(S)}$. In iteration
$i$, we find the unique minimal set $S_i \subset S_{i-1}$ that
minimizes $\frac{|S_{i-1}|-|S_i|}{f(S_{i-1}) - f(S_i)}$. Notice that
$S_k \subset S_{k-1} \subset ... \subset S_1 \subset V$. We say the
relative density of $\hat{S_i} = S_{i-1}\setminus S_i$ is $\lambda_i
= \frac{|S_{i-1}|-|S_i|}{f(S_{i-1}) - f(S_i)}$.  For $u\in \hat{S}_i$,
we say the density of $u$ is $\lambda_u = \lambda_i$. Hence the
dense decomposition of $f$ is $\hat{S}_1, ..., \hat{S}_k$ with
densities $\lambda_1, \ldots, \lambda_k$.
We refer to this decomposition as the first variant which is based on
deletions. 

We now describe a second dense decomposition for submodular functions.
Let $f:2^V \rightarrow \mathbb{R}_+$ be a monotone submodular function.
Consider the supermodular function $g:2^V \rightarrow \mathbb{R}_+$
where $g(X) = f(V) - f(V \setminus X)$ for all $X \subseteq V$.
From \lemref{submodularnegissuper}, $g$ is monotone supermodular.
We can then apply Corollary~\ref{unique-decomp} to obtain a dense decomposition
of $g$. Let $T_1, T_2, \ldots, T_{k'}$ be the unique decomposition
obtained by considering $g$ and let $\hat{\lambda}_1, ...,
\hat{\lambda}_{k'}$ be the corresponding densities. Note that this
second decomposition is based on contractions.

Not too surprisingly, the two decompositions coincide, as we show in
the next theorem. The main reason to consider them separately is for
technical ease in applications where one or the other view is more
natural.

\begin{theorem}{[Proof of \apdxref{pf:decompositions-same}]}
\thmlab{decompositions-same}
Let $\hat{S}_1, ..., \hat{S}_k$ be a dense decomposition (using deletion variant) of a submodular function $f$ with densities $\lambda_i, \ldots, \lambda_k$. Let $T_1, ..., T_{k'}$ be a dense decomposition (using contraction variant) of the same function with densities $\hat{\lambda}_1, ..., \hat{\lambda}_{k'}$. 
We have (i) $k'=k$, (ii) $\hat{S}_1, ...\hat{S}_k$ is exactly $T_1, ..., T_k$, and (iii) $\hat{\lambda}_i = \frac{1}{\lambda_{i}}$ for $1 \le i \le k$.
\end{theorem}

\subsection{Fujishige's results on lexicographically optimal bases}
Fujishige \cite{fujishige} gave a polyhedral view of the dense
decomposition which is the central ingredient in our work. He stated
his theorem for polymatroids, however, it can be easily generalized to
contrapolymatroids. We restrict attention to the unweighted case for
notational ease --- \cite{fujishige} treats the weighted case.

Vectors in $\mathbb{R}^n$ can be totally ordered by sorting the
coordinates in increasing order and considering the lexicographical ordering of the
two sorted sequences of length $n$. In the following, for $a, b \in \mathbb{R}^n$ we use
$a \prec b$ and $a \preceq b$ to refer to this order. We say
that a vector $x$ in a set $D$ is lexicographically minimum (maximum)
if for all $y \in D$ we have $x \preceq y$ ($y \preceq x$).

Fujishige proved the following theorem for polymatroids.
\begin{theorem}[\cite{fujishige}]
\label{fujishige-quad}
  Let $f:2^V \rightarrow \mathbb{R}_+$ be a monotone submodular
  function (a polymatroid) and let $B_f$ be its base polytope. Then 
  there is a unique lexicographically maximum base $b^* \in B_f$ and
  for each $v \in V$, $b^*_v = \lambda_v$. Moroever, $b^*$ is the
 optimum solution to the quadratic program: $\min \sum_{v} x_v^2
 \text{~subject to}~ x \in B_f$. 
\end{theorem}

The preceding theorem can be generalized to contrapolymatroids in
a straight forward fashion and this was explicitly pointed out in
\cite{farouk-neurips}. We paraphrase it to be similar to the preceding
theorem statement.

\begin{theorem}
   Let $f:2^V \rightarrow \mathbb{R}_+$ be a monotone supermodular
  function (a contrapolymatroid) and let $B_f$ be its base polytope. Then 
  there is a unique lexicographically minimum base $b^* \in B_f$ and
  for each $v \in V$, $b^*_v = \lambda_v$. Moreover, $b^*$ is the
 optimum solution to the quadratic program: $\min \sum_{v} x_v^2
 \text{~subject to}~ x \in B_f$.
\end{theorem}

\subsection{Approximating a lexicographically optimal base using Frank-Wolfe}
\label{fw-lex-opt-base-sec}
Consider the convex quadratic program
$\min \sum_{v \in V} x_v^2 \text{~subject to~} x \in B_f$ where $B_f$
is the base polytope of $f$ (could be submodular of supermodular).  We
can use the Frank-Wolfe method to approximately solve this
optimization problem.  The gradient of the quadratic function is
$2x$ and it follows that in each iteration, we need to answer the
linear minimization oracle of
$\text{LMO}(w)=\argmin_{{\bf s} \in B_f}\langle {\bf s} , 2{\bf w}
\rangle $ for ${\bf w} \in B_f$. This is equivalent to
$\argmin_{{\bf s} \in B_f}\langle {\bf s} , {\bf w} \rangle$, in other
words optimizing a linear objective over the base polytope.
Edmonds \cite{edmonds} showed that the simple greedy algorithm
is an $O(|V|\log |V|)$ time exact algorithm (assuming $O(1)$ time oracle access to $f$).
\begin{theorem}{\cite{edmonds}}
    Fix a polymatroid $f:2^V\rightarrow \Re_+$. Given ${\bf
      w}\in B_f$, sort $V=\{v_1, ..., v_n\}$ in descending order of
    ${\bf w}_i$ into $\{s_1, ..., s_n\}$. Let $A_i=\{s_1, ...,
    s_{i}\}$ for $1\leq i\leq n$ with $A_0=\emptyset$. Define ${\bf
      s}^\ast_i = f(A_i)-f(A_{i-1})$. Then ${\bf s}^\ast =
    \argmin_{{\bf s} \in B_f}\langle {\bf s}, {\bf w} \rangle $.
\end{theorem}
The theorem also holds for supermodular functions but by reversing the order from descending to ascending order of ${\bf w}$ and complimenting the set $A_i$. 
\begin{theorem}{\cite{edmonds}}
\label{supermodularedmonds}
    Fix a contrapolymatroid $f:2^V\rightarrow \Re_+$. Given ${\bf w}\in B_f$, sort $V=\{v_1, ..., v_n\}$ in ascending order of ${\bf w}_i$ into $\{s_1, ..., s_n\}$. Let $A_i=\{s_i, ..., s_{n}\}$ for $1\leq i\leq n$ with $A_{n+1}=\emptyset$. Define ${\bf s}^\ast_i = f(A_i)-f(A_{i+1})$. Then ${\bf s}^\ast = \argmin_{{\bf s} \in B_f}\langle {\bf s}, {\bf w} \rangle $.
\end{theorem}

Both algorithms are dominated by the sorting step and thus takes $O(|V|\log |V|)$ time. 
These simple algorithms imply that the Frank-Wolfe algorithm can be used on the 
quadratic program to obtain an approximation to the lexicographically maximum (respectively
minimum) bases of submodular (respectively supermodular)
functions. The standard Frank-Wolfe algorithm would need
$O(\frac{\text{diam}(B_f)^2}{\epsilon^2})$ iterations to converge to a
vector $\hat{b}$ satisfying $\normX{\hat{b}-b^\ast}\leq \epsilon$.

\section{Application 1: Convergence of \textsc{Greedy++} and
  \textsc{Super-Greedy++}}
We begin by describing $\gplus$ from \cite{flowless} and its
generlization $\sgplus$ \cite{chandra-soda}. $\gplus$ is built upon
a modification
of the peeling idea of $\greedy$, and applies it over several iterations. The algorithm
initializes a weight/load on each $v \in V$, denoted by $w(v)$, to
$0$. In each iteration it creates an ordering by peeling the vertices: the next vertex to
be chosen is  $\argmin_v (w(v) + \text{deg}_{G'}(v))$ where $G'$
is the current graph (after removing the previously peeled
vertices). At the end of the iteration, $w(v)$ is increased by the
degree of $v$ when it was peeled in the current iteration. A
precise description can be found below. $\sgreedy$ is a natural generalization of $\greedy$ to supermodular functions, and $\sgplus$ generalizes $\gplus$.  A formal description of the algorithm is given below.

\noindent
\begin{minipage}[t]{0.46\textwidth}
\begin{algorithm}[H]
  \begin{algorithmic}
    \State Initialize $w(u)\leftarrow 0$ for all $u\in V$
    \State $G^\ast \leftarrow G$
    \For {$k\leftarrow 0$ to $T-1$}
        \State $G' \leftarrow G$
        \While {$|G'|>1$}
            \State $u \leftarrow \argmin\limits_{u\in G'}(w(u) + deg_{G'}(u))$
            \State $w(u) \leftarrow w(u) + deg_{G'}(u)$
            \State $G' \leftarrow G' - \{u\}$
            \If{ $\lambda(G')>\lambda(G^\ast)$}
                \State $G^\ast \leftarrow G'$
            \EndIf{}
        \EndWhile{}
    \EndFor{}
    \State 
  \Return $G^\ast$
 \end{algorithmic}
  \caption{\textsc{Greedy++($G(V,E), T$)} \cite{flowless}}\label{GreedyppOriginal}
 \end{algorithm}
\end{minipage}
\hfill
\begin{minipage}[t]{0.51\textwidth}
\begin{algorithm}[H]
  \begin{algorithmic}
    \State Initialize $w(u)\leftarrow 0$ for all $u\in V$
    \State $S^\ast \leftarrow V$
    \For {$k\leftarrow 0$ to $T-1$}
        \State $V' \leftarrow V$
        \While {$|V'|>1$}
            \State $u \leftarrow \argmin\limits_{u\in V'}\{ w(u) + f(V')-f(V'-u) \}$
            \State $w(u) \leftarrow w(u) + f(V')-f(V'-u) $
            \State $V' \leftarrow V' - u$
            \If{ $\frac{f(V')}{|V'|}>\frac{f(S^\ast)}{|S^\ast|}$}
                \State $S^\ast \leftarrow V'$
            \EndIf{}
        \EndWhile{}
    \EndFor{}
    \State 
  \Return $S^\ast$
 \end{algorithmic}
  \caption{\textsc{Super-Greedy++( $f, T$)} \cite{chandra-soda}}\label{SuperGreedypp}
 \end{algorithm}
\end{minipage}

\medskip

The goal of this section is to prove
Theorem~\ref{thm:intro-greedyplusplus} on the convergence of $\sgplus$
and $\gplus$ to the lexicographically maximal base. 

\subsection{Intuition and main technical lemmas}
As we saw in
Section~\ref{fw-lex-opt-base-sec}, if one applies the Frank-Wolfe algorithm to solve the
qaudratic program $\min \sum_{v \in V} x_v^2 \text{~subject to~} x \in
B_f$, each iteration corresponds to finding a minimum weight base of
$f$ where the weights are given by the current vector $x$. Finding a
minimum weight base corresponds to sorting $V$ by $x$. However,
$\sgplus$ and $\gplus$ use a more involved peeling algorithm in each
iteration; the peeling is based on the weights as well as the degrees
of the vertices and it is not a static ordering (the degrees change as
peeling proceeds). This is what makes it non-trivial to formally analyze
these algorithms. In \cite{chandra-soda}, the authors used a
connection to the multiplicative weight update method via LP
relaxations. Here we rely on the quadratic program and noisy
Frank-Wolfe. The high-level intuition, that originates in \cite{chandra-soda}, is the following. As the
algorithm proceeds in iterations, the weights on the vertices
accumulate; recall that the total increase in the 
weight in the case of \dsg is $m = |E|$. The degree term, which influences the
peeling, is dominant in early iterations, but its influence on the ordering of the
vertices decreases eventually as the weights of the vertices get larger. 
It is then plausible to conjecture
that the algorithm behaves like the standard Frank-Wolfe method in the
limit. The main question is how to make this intuition precise. \cite{chandra-soda} relies on a connection to the MWU method while we use a connection to noisy Frank-Wolfe.

For this purpose, consider an iteration of $\gplus$ and
$\sgplus$. The algorithm peels based on the
current weight vector and the degrees. We isolate and abstract
this peeling algorithm and refer to it as Weighted-Greedy and
Weighted-SuperGreedy respectively, and formally describe them with
the weight vector $w$ as a parameter. 

\noindent
\begin{minipage}[t]{0.49\textwidth}
\begin{algorithm}[H]
Input: $G(V,E)$ and $w(u)$ for $u\in V$ 
  \begin{algorithmic}
    \State $G' \leftarrow G$
    \State Initialize $\hat{d}(u)=0$ for all $u\in V$.
    \While {$|G'|>1$}
         \State $u \leftarrow \argmin_{u\in G'}(w(u) + deg_{G'}(u))$
        \State $\hat{d}(u) \leftarrow deg_{G'}(u)$
        \State $G' \leftarrow G' - \{u\}$
    \EndWhile{}
  \Return $\hat{d}$
 \end{algorithmic}
  \caption{\textsc{Weighted-Greedy($G$, $w$)}} \label{weighted-greedypp} \label{WGPP}
 \end{algorithm}
\end{minipage}
\hfill
\begin{minipage}[t]{0.49\textwidth}
\begin{algorithm}[H]
Input: Supermodular $f:2^V \to \Re_+$, $w(u)$ for $ u\in V$ 
  \begin{algorithmic}
    \State $V' \leftarrow V$
    \State Initialize $\hat{d}(u)=0$ for all $u\in V$.
    \While {$|V'|>1$}
        \State $u \leftarrow \argmin\limits_{u\in G'}(w(u) + f(V')-f(V'-u)$
        \State $\hat{d}(u) \leftarrow   f(V')-f(V'-u)$ 
        \State $V' \leftarrow V' - u$
    \EndWhile{}
  \Return $\hat{d}$
 \end{algorithmic}
  \caption{\footnotesize \textsc{Weighted-SuperGreedy($f$, $w$)}} \label{weighted-supergreedypp} \label{WSGPP}
 \end{algorithm}
\end{minipage}

\medskip

The peeling algorithms also compute a base $\hat{d} \in B_f$. In the case of graphs and \dsg, $\hat{d}(u)$ is set to the degree of the vertex $u$ when it is peeled. One can alternatively view the base as an orientation of the edges of $E$.  
Define for each edge $uv\in G$ two weights $x_{uv}, x_{vu}$. We say that $\bf x$ is \emph{valid} if $x_{uv}+x_{vu}=1$ and $x_{uv}, x_{vu}\geq 0$ for all $\{u,v\}\in E(G)$. For $b\in \Re^{|V|}$, we say $x$ \emph{induces} $b$ if $b_u=\sum_{v\in \delta(u)}x_{uv}$ for all $u\in V$. We say a vector $d$ is an \emph{orientation} if there is a valid $x$ that induces it. 
\begin{lemma}[\cite{farouk-neurips}] 
\lemlab{characterization-dsg}
For $f(S)=\cardin{E(S)}$, $b\in B_f$ if and only if $b$ is an orientation. 
\end{lemma}

Recall that the Frank-Wolfe algorithm, for a given weight vector $w: V
\rightarrow \mathbb{R}_+$, computes the minimum-weight base $b$
with respect to $w$ since $\langle w, b \rangle = \min_{y \in B_f}
\langle w, y\rangle$. It is worth taking a moment to note that this base (or orientation due to \lemref{characterization-dsg}) is easily computable: we orient each edge integrally (i.e $x_{vu}=1, x_{uv}=0$) from $v$ to $u$ if $w(u)\geq w(v)$, and the other way otherwise. A simple exchange argument yields a proof of correctness and is implicit in many works \cite{frankwolfe}. This induces an optimal base $d^\ast_w$ with respect to $w$. Our goal is to compare how the peeling order created by Weighted-Greedy (and Weighted-SuperGreedy) compares
with the best base. The following two
technical lemmas formalize the key idea. The first is tailored to \dsg and the
second applies to \dss.

\begin{lemma}{[Proof in \apdxref{pf:approx}]}
\lemlab{approx}
Let $\hat{d}$ be the output from \textsc{Weighted-Greedy}$(G, w)$ and
$d^\ast_w$ be the optimal orientation with respect to $w$. Then
$\langle w, \hat{d} \rangle \leq \langle w, d^\ast_w \rangle +
\sum_{u}deg_G(u)^2$. In particular, the
additive error \textbf{does not depend} on the weight vector $w$.
\end{lemma}

\begin{lemma}{[Proof in \apdxref{pf:approxDSS}]}
\lemlab{approxDSS}
For a supermodular function $f:2^V \to \Re_+$, let $\hat{d}$ be the output from \textsc{Weighted-SuperGreedy}$(f, w)$ (Algorithm \ref{weighted-supergreedypp}) and $d^\ast_w$ be the optimal vector with respect to $w$ as described in Theorem \ref{supermodularedmonds}. Then $\langle w, \hat{d}  \rangle \leq \langle w, d^\ast_w \rangle + n\sum_{u\in V}f(u | V-u)^2$. In particular, the additive error \textbf{does not depend} on the weight vector $w$.
\end{lemma}

\subsection{Convergence proof for $\gplus$}

Why is \lemref{approx} crucial? First, observe that the minimizer
$d^\ast_w$ of $\langle w, d\rangle$ is exactly the same minimizer as
$\langle Kw, d\rangle$ for any constant $K>0$ (and vice-versa).

\begin{lemma}
\label{scaleapprox}
Let $\hat{d}_K$ be the output of \textsc{Weighted-Greedy}$(G, Kw)$. Then $\langle w, \hat{d}_K \rangle \leq \langle w, d^\ast_w \rangle + \frac{\sum_u deg_G(u)^2}{K}$.
\end{lemma}
\begin{proof}
By \lemref{approx}, $\sum_{u\in V} Kw(u)\hat{d}_K(u) \leq \min\limits_{\text{orientation }d}\left( \sum_{u\in V} Kw(u)d(u) \right) + \sum_u deg_G(u)^2
$. 
Dividing by $K$ implies the claim.
\end{proof}

We are now ready to view \textsc{Greedy++}
as a noisy Frank-Wolfe algorithm. Algorithm \ref{greedypp} shows how
\textsc{Greedy++} could be interpreted.

\begin{algorithm}
Input: $G=(V,E)$ and $w(u)$ for $u\in V$ 
  \begin{algorithmic}
    \State Initialize $b^{(0)}\leftarrow \textsc{Weighted-Greedy}(G, {\bf 0})$ \Comment{$b^{(0)}$ is a valid orientation}
    \For {$k\leftarrow 0$ to $T-1$} 
        \State $\gamma \leftarrow \frac{1}{k+1}$
        \State $d^{(k+1)} \leftarrow \textsc{Weighted-Greedy}(G, (k+1)b^{(k)})$
        \State $b^{(k+1)}\leftarrow (1-\gamma)b^{(k)}+\gamma d^{(k+1)}$
    \EndFor{}
  \Return $b^{(T)}$
 \end{algorithmic}
  \caption{\textsc{Greedy++($G(V,E)$)}} \label{greedypp}
\end{algorithm}

The algorithm is exactly the same as the one described in Algorithm
\ref{GreedyppOriginal}. Indeed, one can prove that $kb^{(k)}$ is
precisely the weights that \textsc{Greedy++} ends with at round $k$ by
induction. Observe that $(k+1)b^{(k+1)} = kb^{(k)} + d^{(k+1)}$ which
is precisely the load as described in Algorithm \ref{GreedyppOriginal} (via induction). We note that $\gamma\leftarrow 1/(k+1)$ is crucial here to ensure we are taking the average.  \lemref{greedyppgoodelmo} in the appendix implies that
each peel in Algorithm \ref{GreedyppOriginal} is
$\frac{\delta C_f}{k+2}$-approximate linear minimization oracle. Using
\thmref{FW-Resistant}, this implies that \textsc{Greedy++} (as
described in Algorithm \ref{GreedyppOriginal}) converges to $b^\ast$ in $\Tilde{O}(\frac{mn^2}{\epsilon^2})$ iterations since $\delta=O(\frac{\sum_u d_G(u)^2}{m})$ and $C_f = O(\sum_u d_G(u)^2)$. We use the probabilistic method to bound $C_f$ in the Appendix. 

\mypara{Extension to $\sgplus$:} An essentially similar analysis works for $\sgplus$. Instead of \lemref{approx}, we rely on \lemref{approxDSS}. For technical reasons, the convergence analysis of \sgplus is slightly weaker than for \gplus.

\section{Application 2: Greedy Tree Packing interpreted
  via Frank-Wolfe}
  \label{sec:treepacking}
Let $G=(V,E)$ be a graph with non-negative edge capacities.
The goal of this section is to view Thorup's definitions of ideal edge loads and the associated tree packing from a different perspective, and to derive an alternate convergence analysis of his greedy tree packing algorithm~\cite{Thorup07,Thorup08}.  In previous work, Chekuri, Quanrud and Xu \cite{ChekuriQX20} obtained a
different tree packing based on an LP relaxation for
$k$-cut, and used it in place of ideal tree packing. Despite this, a proper understanding of Thorup's ideas was not clear. We address this gap.

We restrict our attention to unweighted multi-graphs throughout this
section, and comment on the capacitated case at the end of the
section. Let $G=(V,E)$ be a connected multi-graph, with 
 $n$ vertices and $m$ edges. Consider the graphic matroid
$\mathcal{M}_G(E, \mathcal{F})$ induced by $G$; $E$ is the ground set,
and $\mathcal{F}$ consists of all sub-forests of $G$. The bases of the matroid are
precisely the spanning trees of $G$. Consider the rank function
$r: 2^E \rightarrow \mathbb{Z}_+$ of $\mathcal{M}_G$. $r$ is
submodular, and it is well-known that for a edge subset
$X \subseteq E$, $r(X) = n - \kappa(X)$ where $\kappa(X)$ is the
number of connected components induced by $X$.

\subsection{Thorup's recursive algorithm as dense decomposition}
For consistency with previous notation, we use $f$ to denote the submodular rank
function $r$. We first describe ideal loads as defined by Thorup.
Consider the Tutte--Nash-Williams partition $P$ for $G$. Recall that $P$
minimizes the ratio $\frac{|E(P)|}{|P|-1}$ among
all partitions, and this ratio is $\tau(G)$. For each edge
$e \in E(P)$, assign $\ell^*(e) = \frac{1}{\tau(G)}$. Remove the edges
in $E(P)$ to obtain a graph $G'$ which now consists of several
disconnected components. Recursively compute ideal loads for the edges
in each component of $G'$ (the process stops when $G$ has no edges).

We claim that Thorup's recursive decomposition coincides with the 
dense decomposition of $f$ (the first variant). To see this, it suffices
to see the first step of the dense decomposition. We find the minimal set
$S_1 \subseteq E$ that minimizes $\frac{|E|-|S|}{f(E)-f(S)}$. 
We let $\hat{S}_1=E\setminus S_1$ and assign the edges in $\hat{S}_1$ the
density $\frac{f(E)-f(S)}{|E|-|S|}$. Then, we ``delete''
$\hat{S}_1$. Observe that $\hat{S}_1=E\setminus S_1$ is just the edges crossing the
partition $P(S_1)$ defined by the $\kappa(S_1)$ connected components
spanned by $S_1$. Also, recall that
$\frac{f(E)-f(S_1)}{|E|-|S_1|}=\frac{\kappa(S_1)-1}{|E\setminus
  S_1|}=\frac{|P(S_1)|-1}{E(P(S_1))}=\frac{1}{\tau(G)}$. Hence, the
density assigned to edges in $\hat{S}_1$ is exactly
$\frac{1}{\tau(G)}$ by the Tutte--Nash-Williams theorem. The next step
is deleting $\hat{S_1}=E\setminus S_1$, which, as discussed above,
are the edges crossing the partition $P(S_1)$.

Via induction we prove the following lemma.
\begin{lemma}
\lemlab{thorup:rec:ok}{[Proof in \apdxref{pf:thorup:rec:ok}]}
    The weights given to the edges by the dense decomposition algorithm on $f$ coincide with $\ell^\ast$. 
\end{lemma}

\subsection{Greedy tree packing converge to ideal relative loads}
Thorup considered the following greedy tree packing algorithm.
For each edge define a load $\ell(e)$ which is initialized to $0$.
The algorithm proceeds in iterations. In iteration $i$ the algorithm
computes an MST $T_i$ in $G$ with respect to edge weights $w(e) =
\ell(e)$. The load of each edge $e \in T_i$ is increased by $1$.
Thorup showed that as $k \rightarrow \infty$, the quantity
$\ell(e)/k$ converges to $\ell^*(e)$ for each edge $e$. His proof is fairly technical. In this section, we present a
different proof of this fact that uses the machinery we have built thus
far. 

\begin{lemma}
\lemlab{lstar-lex-max}
The vector $\ell^\ast$ is the lexicographically maximal base of the spanning tree polytope.  
\end{lemma}
\begin{proof}
We showed that Thorup's algorithm simply runs the dense decomposition
on the graph for the rank function of the graphic matroid induced by
$G$. The bases of the matroid are the spanning trees of $G$ and hence
the base polytope of $f$ is the spanning tree polytope of $G$.
The dense decomposition gives us the lexicographically maximum base
$f$, and hence $\ell^\ast$ is the lexicographically maximal base of
the spanning tree polytope.
\end{proof}

Hence, $\ell^\ast$ is the unique solution to the quadratic program of
minimizing $\sum_{e}\ell(e)^2$ subject to $\ell \in \textsc{SPT}(G)$ where
$\textsc{SPT}(G)$ is the spanning tree base polytope. We can thus
apply a noisy Frank-Wolfe algorithm on the quadratic program to get Algorithm \ref{GreedyTreePackingFW}.

\begin{algorithm}
Input: $G(V,E)$ 
  \begin{algorithmic}
    \State Initialize $l^{(0)}(u) = \mathbb{1}\{e\in T\}$ for any spanning tree $T$.
    \For {$k\leftarrow 0$ to $T-1$} 
        \State $\gamma \leftarrow \frac{1}{k+1}$
        \State $d^{(k+1)} \leftarrow \min\limits_{s \in \textsc{SPT}(G)} \langle l^{(k)},s  \rangle$ \Comment{This is the minimum spanning tree with respect to $l^{(k)}$}
        \State $l^{(k+1)}\leftarrow (1-\gamma)l^{(k)}+\gamma d^{(k+1)}$
    \EndFor{}
  \Return $b^{(T)}$
 \end{algorithmic}
  \caption{\textsc{Frank-Wolfe-Greedy-TreePack($G(V,E)$)}} \label{GreedyTreePackingFW}
 \end{algorithm}

 The main observation is that this algorithm is
 \textbf{exactly} the same as the Thorup's greedy tree packing!
 Indeed, observe that $(k+1)\ell^{(k+1)} \leftarrow k\ell^{(k)}+d^{(k+1)} = k\ell^{(k)} +
 \mathbb{1}\{e\in \textsc{MST}(G, \ell^{(k)})\}$ where
 $\textsc{MST}(G,w)$ is a minimum spanning tree of $G$ with respect to
 edge weights $w$. Since noisy Frank-Wolfe converges, then
 $\ell^{(k)}$ converges to $\ell^\ast(e)$, and greedy tree packing
 converges.

 We now establish the convergence guarantee for greedy tree packing. 
 For the spanning tree polytope of an $m$ edge graph, the curvature constant $C_f \leq 4m$ because for $x,y\in B_f$, $2(x-y)^T(x-y) = \sum_{e\in E} (x_e-y_e)^2 \leq 4m$. Plugging this bound into \thmref{FW-Resistant}, we get that after $k=O(\frac{m  \log (m/\epsilon)}{\epsilon^2})$ iterations, we have $\normX{\ell^{(k)}-\ell^\ast} \leq \epsilon$. 

Suppose we run the standard Frank-Wolfe algorithm with $\gamma=2/(k+2)$. Then, the convergence guarantee improves to $O(\frac{m}{\epsilon^2})$. Note that each iteration still corresponds to finding an MST in the graph with weights. However, the load vector is no longer a simple average of the trees taken so far.

\mypara{Comparison to Thorup's bound and analysis:} Thorup \cite{Thorup08} considered ideal tree packings in capacitated graphs; let $c(e) \ge 1$ (via scaling) denote the capacity of edge $e$. Via \cite{fujishige}, one sees that the optimum solution of the quadratic program $\sum_{e} x_e^2/c(e) \text{~subject to~} x \in SP(G)$ is the ideal load vector $\ell^\ast$. Greedy tree packing generalizes to the capacitated case easily; in each iteration we compute the MST with respect to weights $w(e) = \ell(e) c(e)$. Thorup proved the following.
\begin{theorem}[\cite{Thorup08}]
    Let $G=(V,E)$ be capacitated graph. Greedy tree packing after $O(\frac{m \log (mn/\epsilon)}{\epsilon^3})$ iterations ouputs a load vector $\ell$ such that for each edge $e \in E$, $(1-\eps) \ell^*(e) \le \ell(e) \le (1+\eps)\ell^*(e)$.
\end{theorem}

We observe that if all capacities are $1$ (or identical) then Thorup's guarantee is 
that $|\ell(e) - \ell^\ast(e)| \le O(\eps)$. For this case, via Frank-Wolfe, we obtain the much stronger guarantee that $||\ell - \ell^\ast|| \le \eps$ which easily implies the per edge condition, however the per edge guarantee does not imply a guarantee on the norm. Further, in the unweighted case, our iteration complexity dependence on $\epsilon$ is $1/\eps^2$ while Thorup's is $1/\eps^3$. Thorup's guarantee works for the capacitated case in strongly polynomial number of iterations. We can adapt the Frank-Wolfe analysis for the capacitated case but it would yield a bound that depends on 
$C = \sum_e c(e)$ (in the unweighted case $C = m$); on the other hand the guarantee provided by Frank-Wolfe is stronger. 

It may seem surprising that the same greedy tree packing algorithm yields different types of guarantees based on the type of analysis used. We do not have a completely satisfactory explanation but we point out the following. Thorup's analysis is a non-trivial refinement of the standard MWU type analysis of tree packing \cite{pst-faafpcp-95, Young1995RandomizedRW}. See \cite{AHK-MWU} for an excellent survey on MWU. As already noted in \cite{farouk-neurips}, if one use Frank-Wolfe (with $\gamma = 1/(k+1)$) with the softmax potential function that is standard in MWU framework, then the resulting algorithm would also be greedy tree packing. Fujishige's uses a quadratic objective to guarantee that the optimum solution is the unique maximal base but in fact any increasing strongly convex function would suffice. In the context of optimizing a linear function over $B_f$, due to the optimality of the greedy algorithm for this, the only thing that determines the base is the ordering of the elements of $V$ according to the weight vector; the weights themselves do not matter. Thus, Frank-Wolfe applied to different convex objectives can result in the same greedy tree/base packing algorithm. However, the specific objective can determine the guarantee one obtains after a number of iterations. The softmax objective is better suited for obtaining relative error guarantees while the quadratic objective is better suited for obtaining additive error guarantees. Thorup's analysis is more sophisticated due to the per edge guarantee in the capacitated setting. A unified analysis that explains both the relative and additive guarantees is desirable and we leave this is an interesting direction for future research.

\newpage 

\bibliographystyle{unsrt}
\bibliography{dsg}

\newpage 
\section{Auxiliary Lemmas}

\begin{figure}
    \centering
    \includegraphics[width=0.45\textwidth]{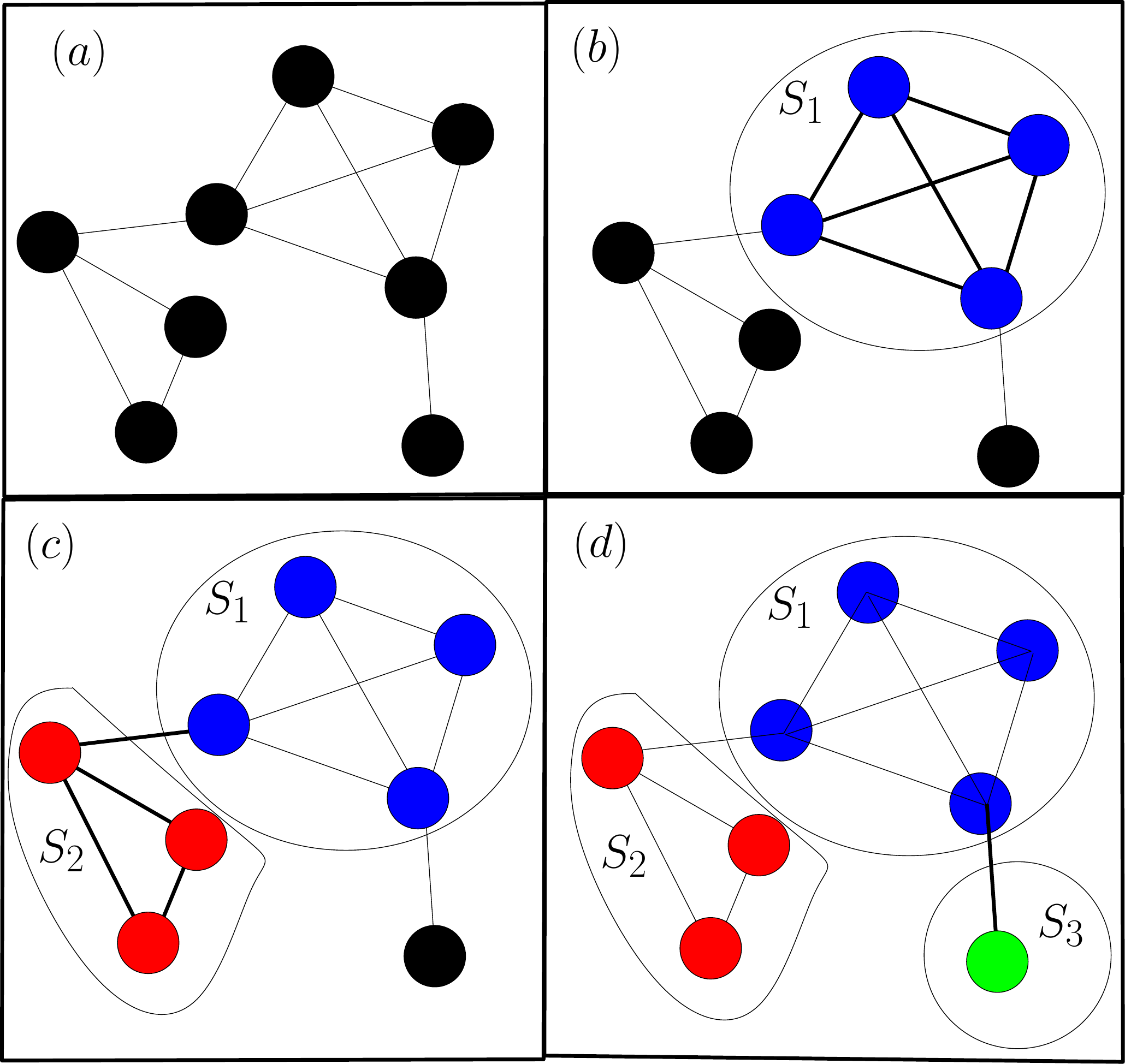}
    \caption{Densest Subgraph Decomposition Example. The densest subgraph $S_1$ is shown in $(b)$ with blue vertices with density $6/4=1.5$. We ``contract'' $S_1$ (the blue vertices) and find the densest subgraph $S_2$ with density $(3+1)/3=4/3$. Finally, we contract $S_2$, and find $S_3$ with density $(0+1)/1=1$. }
    \label{fig:dense_decomposition_example}
\end{figure}

\begin{lemma}
\lemlab{greedyppgoodelmo}
For $k\geq 0$
\[
\sum_{u\in V} b^{(k)}(u)d^{(k+1)}(u) \leq \left( \min\limits_{ \text{orientation} \overrightarrow{d}}\sum_{u\in V}b^{(k)}(u)\overrightarrow{d}(u) \right) + \frac{\delta C_f}{k+2}
\]
\end{lemma}
\begin{proof}
Using Lemma \ref{scaleapprox} using $(k+1)b^{(k)}$ as the weights, we have that 
\[
\sum_{u\in V} b^{(k)}(u)d^{(k+1)}(u) \leq \left( \min\limits_{\text{orientation}\overrightarrow{d}}\sum_{u\in V}b^{(k)}(u)\overrightarrow{d}(u) \right) + \frac{\sum_{u\in V}deg_G(u)^2}{k+1}
\]
We show in \lemref{probabilisticmethod} that $C_f \geq 2m$ using the probabilistic method, and so for $\delta=\Theta(\frac{\sum_u deg_G(u)^2}{m})$, we have that 
\[
\frac{\sum_{u\in V}deg_G(u)^2}{k+1} \leq \frac{\delta C_f}{k+2}
\]
\end{proof}

\begin{lemma}

\lemlab{probabilisticmethod}
Let $f(b) = \sum_{u}b_u^2$ be the sum of squares of an orientation load vector. Then $2m \leq C_f \leq 2 \sum_u deg_G(u)^2$
\end{lemma}
\begin{proof}
Let $\mathcal{D}$ be the set of valid orientations. Then using the definition of $C_f$ and simplifying, we have that 
\[
C_f = \sup_{x, s \in \mathcal{D}, \gamma \in [0,1]} \frac{2}{\gamma^2} (f(x+\gamma(s-x))-f(x) - \langle \gamma (s-x), 2x \rangle ) = \sup_{x, s \in \mathcal{D}} 2(s-x)^T(s-x)
\]

Let $x,s \in \mathcal{D}$. Clearly $x_u, s_u \leq deg_G(u)$ since they are orientations. So $(s_u-x_u)^2 \leq deg_G(u)^2$. Summing over $u$ establishes the upper bound. 

For the lower bound, we use the probabilistic method. Arbitrarily orient any edge $(i,j)$ with probability $1/2$ towards $i$ and $1/2$ towards $j$. This induces a load $b^1$. Repeat this \textbf{independently} to induce load $b^2$. Note that for any vertex $u$, $b^1_u, b^2_u \sim Bin(deg_G(u), \frac{1}{2})$ and are independent. The variance of $b^1_u - b^2_u$ is thus $(deg_G(u)+deg_G(u))\times 1/2 \times 1/2 = \frac{1}{2}deg_G(u)$. Hence, we have that 
\[
\frac{1}{2}deg_G(u) = \Ex{ (b^1_u - b^2_u)^2 } - \Ex{(b^1_u - b^2_u)}^2 = \Ex{ (b^1_u - b^2_u)^2 }
\]
So we have that 
\[
\Ex{2(b^1-b^2)^T(b^1-b^2)} = 2\sum_{u\in V} \Ex{ (b^1_u - b^2_u)^2 } = 2m
\]
So there must be a realization with $2(b^1-b^2)^T(b^1-b^2)$ at least the expectation value $2m$, and hence the supremum is at least $2m$.

\end{proof}

\section{Proofs}

\subsection{Proof of \thmref{FW-Resistant}}
\apdxlab{pf:FW-Resistant}
\begin{proof}
  We will note that this proof is slight adaptation of Jaggi's proof
  \cite{pmlr-v28-jaggi13}, and is included here for the sake of
  completeness.

Let ${\bf \hat{d}}^{(k+1)}$ be the direction returned by the
approximate linear minimization oracle in iteration $k$. For any
$\gamma \in [0,1]$, from the definition of the curvature constant
$C_f$, we have that
\begin{equation}
\label{first}
f({\bf b}^{(k+1)}) = f({\bf b}^{(k)} + \gamma {\bf \hat{d}}^{(k+1)}) \leq f({\bf b}^{(k)}) + \gamma \langle {\bf \hat{d}}^{(k+1)}-{\bf b}^{(k)}, \nabla f({\bf b}^{(k)}) \rangle + \frac{\gamma^2}{2}C_f
\end{equation}
Note that the ${\bf \hat{d}}^{(k+1)}$ we use is a $\frac{\delta C_f}{k+2}$-approximate linear minimization oracle, and hence $\langle {\bf \hat{d}}^{(k+1)}, \nabla f({\bf b}^{(k)}) \rangle  \leq \langle {\bf d}^{(k+1)}, \nabla f({\bf b}^{(k)}) \rangle + \frac{\delta C_f}{k+2}$. Combining this with (\ref{first}) and rearranging, we get
\begin{equation}
\label{second}
f({\bf b}^{(k+1)}) \leq f({\bf b}^{(k)}) - \gamma g({\bf b}^{(k)}) + \frac{\gamma^2}{2}C_f(1+\delta)
\end{equation}
Where $g({\bf b}^{(k)})= \langle {\bf d}^{(k+1)} - {\bf b}^{(k)}, \nabla f({\bf b}^{(k)}) \rangle $. Next, denote $h({\bf b}^{(k)})=f({\bf b}^{(k)})-f({\bf b}^\ast)$ for the primal error. Convexity  of $f$ implies  that  the  linearization $f({\bf b}) + \langle {\bf s}-{\bf b}, \nabla f({\bf b}) \rangle $ always lies below the graph of $f$. This implies $g({\bf b}^{(k)}) \geq h({\bf b}^{(k)})$. Combining this with (\ref{second}), we get
\begin{equation}
\label{third}
f({\bf b}^{(k+1)}) \leq f({\bf b}^{(k)}) - \gamma h({\bf b}^{(k)}) + \frac{\gamma^2}{2}C_f(1+\delta)
\end{equation}
Subtracting $f({\bf b}^\ast)$ from both sides of (\ref{third}), we get
\begin{equation}
\label{fourth}
h({\bf b}^{(k+1)}) \leq (1-\gamma) h({\bf b}^{(k)}) + \frac{\gamma^2}{2}C_f(1+\delta)
\end{equation}
Let $C = \frac{1}{2}C_f(1+\delta)$ and $\epsilon_k = h({\bf b }^{(k)})$. Then from (\ref{fourth}) we get the recurrence 
\begin{equation}
\label{fifth}
\epsilon_{k+1} \leq (1-\gamma) \epsilon_k + \gamma^2C
\end{equation}
We claim that $\epsilon_k \leq \frac{4C H_{k+1}}{k+1}$ which implies the theorem. For $k=0$, (\ref{fifth}) with $\gamma=\frac{1}{0+1}$ implies $\epsilon_{0+1}\leq C \leq 4C$. For $k\geq 1$, we want the RHS of (\ref{fifth}) to satisfy
\[
(1-\frac{1}{k+1})\frac{4CH_{k+1}}{k+1} + \frac{C}{(k+1)^2} \leq \frac{4CH_{k+2}}{k+2}
\]
Alternatively, this is the same as $4(k+2)H_{k+1}+k(3k+4)\geq 0$ after rearranging, which is satisfied for $k\geq 1$. 
\end{proof}

\subsection{Proof of \lemref{submodularnegissuper}}
\apdxlab{pf:submodularnegissuper}
\begin{proof}
We will show that $h(X)=f(V\setminus X)$ is submodular for submodular $f$. This would imply $-f(V\setminus X)$ is supermodular, and since $f(V)$ is modular, then $g$ is supermodular. 

Let $A, B \subseteq V$. To see why $h$ is submodular, we have the inequalities:
\[
h(A)+h(B) = f(V\setminus A) + f(V\setminus B) \geq f(V\setminus A \cup V\setminus B) + f(V\setminus A \cap V\setminus B)
\]
Note that $V\setminus A \cup V\setminus B = V\setminus (A \cap B)$. In addition, $V\setminus A \cap V\setminus B = V\setminus (A\cup B)$. Hence
\[
h(A)+h(B) \geq  f(V \setminus (A\cup B)) + f(V \setminus (A\cap B)) = h(A\cup B) + h(A\cap B)
\]
\end{proof}

\subsection{Proof of \lemref{sprmod:unqqq}}
\apdxlab{pf:sprmod:unqqq}
\begin{proof}
    Let $S_1, S_2\subseteq V$ be maximal sets achieving the maximum $\lambda = \frac{f(S)}{\cardin{S}}$. Then we have by supermodularity 
    \[
    \frac{f(S_1 \cup S_2)}{\cardin{S_1 \cup S_2}} = \frac{f(S_1 \cup S_2)}{\cardin{S_1}+\cardin{S_2} - \cardin{S_1\cap S_2}} \geq \frac{f(S_1)+f(S_2)-f(S_1\cap S_2)}{\cardin{S_1}+\cardin{S_2} - \cardin{S_1\cap S_2}}
    \]
    Note that $f(S_1\cap S_2)\leq \lambda \cardin{S_1\cap S_2}$ by optimality of $\lambda$ which implies the continued chain  
    \[
    \frac{f(S_1 \cup S_2)}{\cardin{S_1 \cup S_2}} \geq \frac{f(S_1)+f(S_2)-\lambda |S_1\cap S_2|}{\cardin{S_1}+\cardin{S_2} - \cardin{S_1\cap S_2}} = \frac{\lambda \cardin{S_1}+  \lambda \cardin{S_2} - \lambda \cardin{S_1 \cap S_2}}{\cardin{S_1}+\cardin{S_2} - \cardin{S_1\cap S_2}} = \lambda 
    \]
By optimality of $\lambda$, $f(S_1\cup S_2) = \lambda \cardin{S_1\cup S_2}$. By maximality of $S_1, S_2$, $S_1=S_1\cup S_2 = S_2$. 
\end{proof}

\subsection{Proof of \lemref{minimality:submodular}}
\apdxlab{pf:minimality:submodular}
\begin{proof}
    Let $S_1, S_2 \subseteq V$ be minimal sets that minimizes the ratio $ (|V|-|S|)/(f(V)-f(S))$ with value $\lambda$. Then by supermodularity of $g(S)=f(V)-f(S)$, we have
    \[
    \frac{|V|-|S_1\cap S_2|}{f(V)-f(S_1\cap S_2)} = \frac{|V|-|S_1\cap S_2|}{g(S_1\cap S_2)} \leq \frac{|V|-|S_1\cap S_2|}{g(S_1)+g(S_2)-g(S_1\cup S_2)}
    \]
    Note that $|V|-|S_1\cap S_2| = |V|-|S_1|+|V|-|S_2| - |V|+|S_1\cup S_2|$. In addition $\lambda g(S_1\cup S_2)\leq  |V|-|S_1\cup S_2|$ by optimality of $\lambda$. Then we get the chain of inequalities
    \[
    \frac{|V|-|S_1\cap S_2|}{f(V)-f(S_1\cap S_2)} \leq \frac{\lambda g(S_1)+\lambda g(S_2)-\lambda g(S_1\cup S_2)}{g(S_1)+g(S_2)-g(S_1\cup S_2)} = \lambda
    \]
    By optimality of $\lambda$, $(|V|-|S_1\cap S_2|)/(f(V)-f(S_1\cap S_2)) = \lambda$. By minimality of $S_1, S_2$, $S_1=S_1\cap S_2 = S_2$. 
\end{proof}

\subsection{Proof of \thmref{decompositions-same}}
\apdxlab{pf:decompositions-same}
\begin{proof}
\begin{figure}
    \centering
\includegraphics[width=0.4\textwidth]{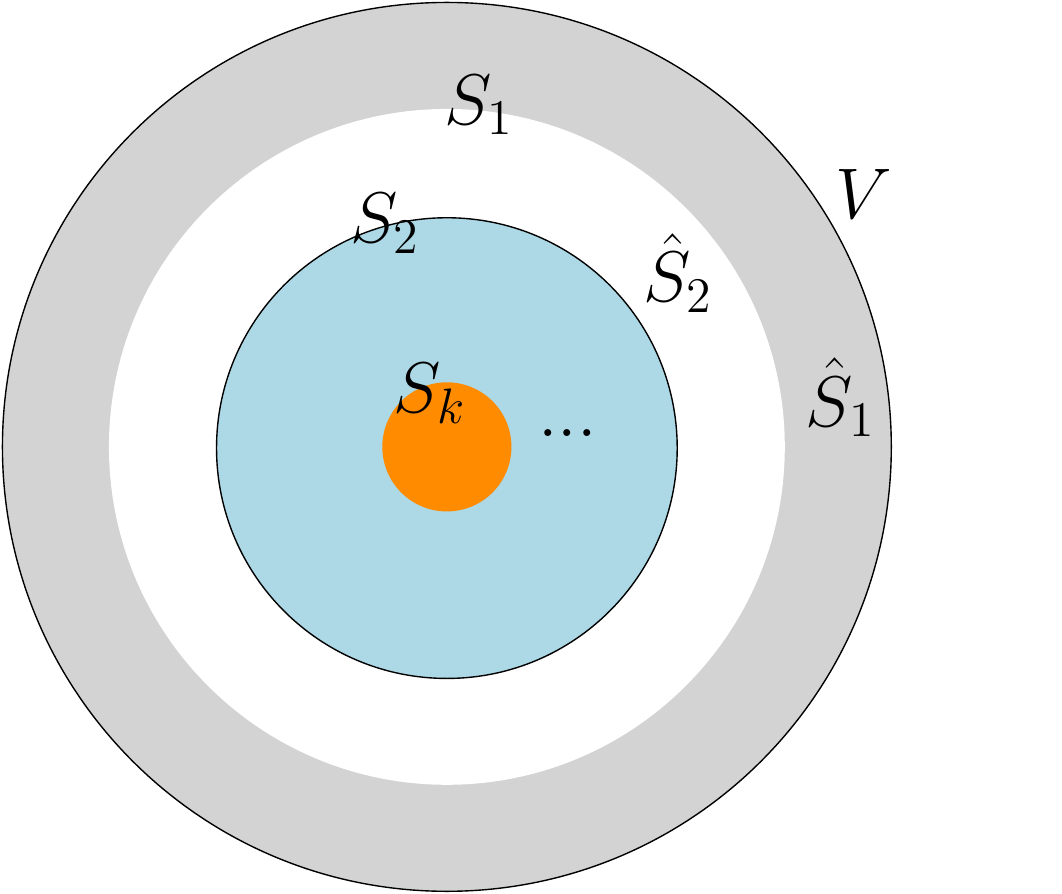}
    \caption{Contraction based decomposition for a submodular function $f$. Shown is the Venn-diagram of $S_1, ..., S_k$ and $\hat{S}_1, ..., \hat{S}_k$}
    \label{fig:decomposition_proof}
\end{figure}
See Figure \ref{fig:decomposition_proof} throughout this proof. 
Note that $\hat{S}_1, ..., \hat{S}_k$ are associated with densities
\[
\frac{|V|-|S_1|}{f(V)-f(S_1)} < \frac{|S_1|-|S_2|}{f(S_1)-f(S_2)} < ... <\frac{|S_{k-1}|-|S_k|}{f(S_{k-1})-f(S_k)}
\]
Similarly, $T_1, ..., T_{k'}$ are associated with densities 
\[
\frac{f(V)-f(V\setminus T_1) }{|T_1|} > \frac{f(V\setminus T_1)-f(V\setminus T_1 \setminus T_2)}{|T_2|} > ... > \frac{f(T_{k'})-f(\phi)}{|T_{k'}|}
\]

We prove the claim by induction. Suppose $T_1, ..., T_{i-1}$ is the same as $\hat{S}_1, ..., \hat{S}_{i-1}$ with the trivial base case. 

First note that $T_i \subseteq S_{i-1}$ since $T_1\cup ... \cup T_{i-1} = (V\setminus S_1)\cup (S_1 \setminus S_2) \cup ... (S_{i-2}\setminus S_{i-1}) = V\setminus S_{i-1}$ and the disjointedness of $\{T_j\}$. Now observe that by the optimality of $S_i$ that 
\[
\frac{|S_{i-1}|-|S_i|}{f(S_{i-1})-f(S_i)} \leq \frac{|S_{i-1}|-|S_{i-1}\setminus T_i|}{f(S_{i-1})-f(S_{i-1}\setminus T_i)} = \frac{|T_i|}{f(S_{i-1})-f(S_{i-1}\setminus T_i)} = \frac{|T_i|}{f(V\setminus \bigcup_{j<i}T_j)-f(V\setminus \bigcup_{j\leq i}T_j)}
\]
Conversely, by the optimality of $T_i$
\[
\frac{f(V\setminus \bigcup_{j<i}T_j)-f(V\setminus \bigcup_{j\leq i}T_j)}{|T_i|} \geq \frac{f(V\setminus \bigcup_{j<i}T_j)-f(V\setminus \bigcup_{j< i}T_j\setminus (S_{i-1}\setminus S_i))}{|S_{i-1}\setminus S_i|} = \frac{f(S_{i-1})-f(S_{i})}{|S_{i-1}|-|S_i|}
\]
Hence $\lambda_{i} = \hat{\lambda}_i$. This also forces $S_i=S_{i-1}\setminus T_i$ or $T_i=S_{i-1}\setminus S_i = \hat{S}_i$. 
\end{proof}

\subsection{Proof of \lemref{approx}}
\apdxlab{pf:approx}
\begin{proof}
Consider the optimal orientation of the edges $d^\ast_w$. How does reversing one edge (from $(u, v)$ to $(v,u)$) affect the cost of minimization oracle? The out degree of $u$ decreases by $1$ and the out degree of $v$ increases by $1$, and so the objective function \textbf{increases} by $w_v - w_u$. 

Suppose \textsc{Weighted-Greedy++} peels the vertices $u_1, ..., u_n$ in this order. We proceed by induction. Consider all the ``wrongly'' oriented edges $W(u_1)=\{(u_1, u_i) : i>1, w(u_1)>w(u_i) \}$. These edges increase the objective function from the optimal solution value by $\sum_{v\in W(u)}(w(u_1)-w(v))
$. But recall that \textsc{Weighted-Greedy++} chooses $u_1$ because $w(u_1) + deg_{G'}(u_1) \leq w(v) + deg_{G'}(v)
$ for all $v\in W(u)$. Which means that 
\[
w(u_1) - w(v) \leq \deg_{G'}(v) - deg_{G'}(u_1) \leq deg_G(v)
\]
Proceeding by induction on $u_2, ..., u_n$, the ``wrongly'' oriented edges contribute a total of at most $\sum_{u}deg_G(u)^2$ additive error with respect to the correct orientation $d^\ast_w$ since each vertex $v$ contributes at most $deg_G(v)$ from all its neighbors.    
\end{proof}

\subsection{Proof of \lemref{approxDSS}}
\apdxlab{pf:approxDSS}
\begin{proof}
\begin{figure}
    \centering\includegraphics[width=0.8\textwidth]{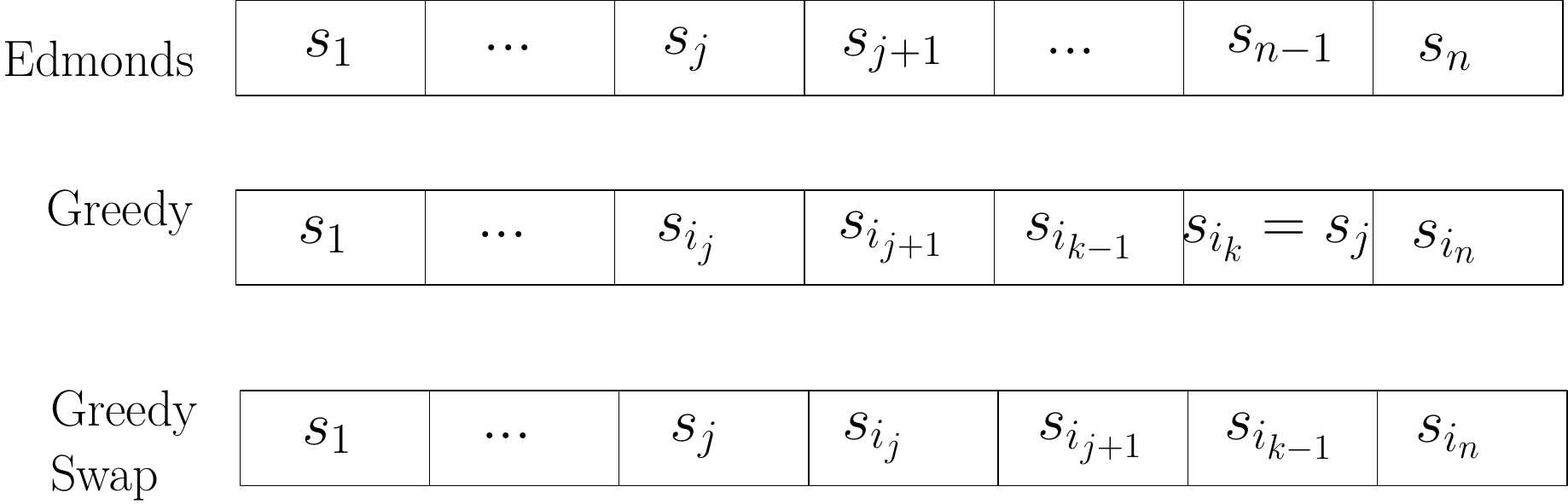}
    \caption{Proof idea of Lemma \lemref{approxDSS}}
    \label{fig:weightedgreedyppProof}
\end{figure}
Suppose \textsc{Weighted-SuperGreedy++} peels $V$ in the order $s_{i_1}, ..., s_{i_n}$, and Edmonds' algorithm peels them in the order $s_1, ..., s_n$ with $w_1\leq ... \leq w_n$. Let $A_j=\{s_{i_{j+1}}, ..., s_{i_n}\}$. Let $j$ be the index of the first disagreement where $s_j \neq s_{i_j}$.  \textsc{Weighted-SuperGreedy++} chooses $s_{i_j}$ over $s_{i_k}$ ($k \geq j$) because 
\[
w(s_{i_j}) + f(s_{i_j}|A_j-s_{i_j}) \leq w(s_{i_k}) + f(s_{i_k}|A_j-s_{i_k})
\]
Which implies by supermodularity
\begin{equation}
\label{weight:difference:bound}
w(s_{i_j}) - w(s_{i_k}) \leq  f(s_{i_k}~|~V-s_{i_k})
\end{equation}
 Suppose $s_{i_k}=s_j$, then we will move $s_{i_j}, s_{i_{j+1}}, ..., s_{i_{k-1}}$ to go after $s_{i_k}=s_j$ in the \textsc{Weighted-SuperGreedy++} order. (see Figure \ref{fig:weightedgreedyppProof}). Swapping consecutive elements $s_{i_a}, s_{i_{a+1}}$ in the \textsc{Weighted-SuperGreedy++} order changes the inner product cost by 
\[
w(s_{i_{a+1}})f(s_{i_{a+1}}|A_{a}+s_{i_a}-s_{i_{a+1}})+w(s_{i_{a}})f(s_{i_{a}}|A_{a+1}) - w(s_{i_a})f(s_{i_{a}}|A_{a})-w(s_{i_{a+1}})f(s_{i_{a+1}}|A_{a+1})
\]
\begin{equation}
\label{difference}
=w(s_{i_a})(f(s_{i_{a}}|A_{a+1}) - f(s_{i_{a}}|A_{a})) - w(s_{i_{a+1}})(f(s_{i_{a+1}}|A_{a+1})-f(s_{i_{a+1}}|A_{a}+s_{i_a}-s_{i_{a+1}}))
\end{equation}
But 
\[
f(s_{i_{a}}|A_{a+1}) - f(s_{i_{a}}|A_{a}) = f(A_{a+1}+s_{i_{a}})-f(A_{a+1}) - f(A_a+s_{i_a})+f(A_a)
\]
And
\[
f(s_{i_{a+1}}|A_{a+1})-f(s_{i_{a+1}}|A_{a}+s_{i_a}-s_{i_{a+1}}) = f(A_{a+1}+s_{i_{a+1}}) - f(A_{a+1}) -f(A_a+s_{i_a}) + f(A_a+s_{i_a} - s_{i_{a+1}})
\]
\[
= f(A_a)-f(A_{a+1}) - f(A_a+s_{i_a}) + f(A_{a+1}+s_{i_a})
\]
And hence both coefficients of $w(s_{i_a}), w({s_{i_{a+1}}})$ in \ref{difference} are equal. Hence by 
(\ref{weight:difference:bound})
\[
(\ref{difference}) = (w_{i_a} - w_{i_{a+1}})(f(s_{i_{a+1}}|A_{a+1})-f(s_{i_{a+1}}|A_{a+1}+s_{i_a}) ) \geq -f(s_{i_{a+1}}|V - s_{i_{a+1}} )^2
\]
Hence, moving all of $s_{i_j}, ..., s_{i_{k-1}}$ to $s_{i_k}=s_j$ \textit{decreases} the cost by at most $\leq (k-j)f(s_{i_{k}}|V)^2 \leq nf(s_{i_{k}}|V-s_{i_{k}})^2$. Summing over all reorderings of the vertices, we see that Edmonds' ordering inner product is at most $n\sum_{u\in V}f(u|V-u)^2$ away from the \textsc{Weighted-SuperGreedy++} order inner product.
For the function $f(S)=|E(S)|$, the bound from \lemref{approx} is better than the bound of Lemma \lemref{approxDSS} by a factor of $n$. We leave improving the bound for future work.  
\end{proof}

\subsection{Proof of \lemref{thorup:rec:ok}}
\apdxlab{pf:thorup:rec:ok}
\begin{proof}
\begin{figure}
    \centering
    \includegraphics[width=0.5\textwidth]{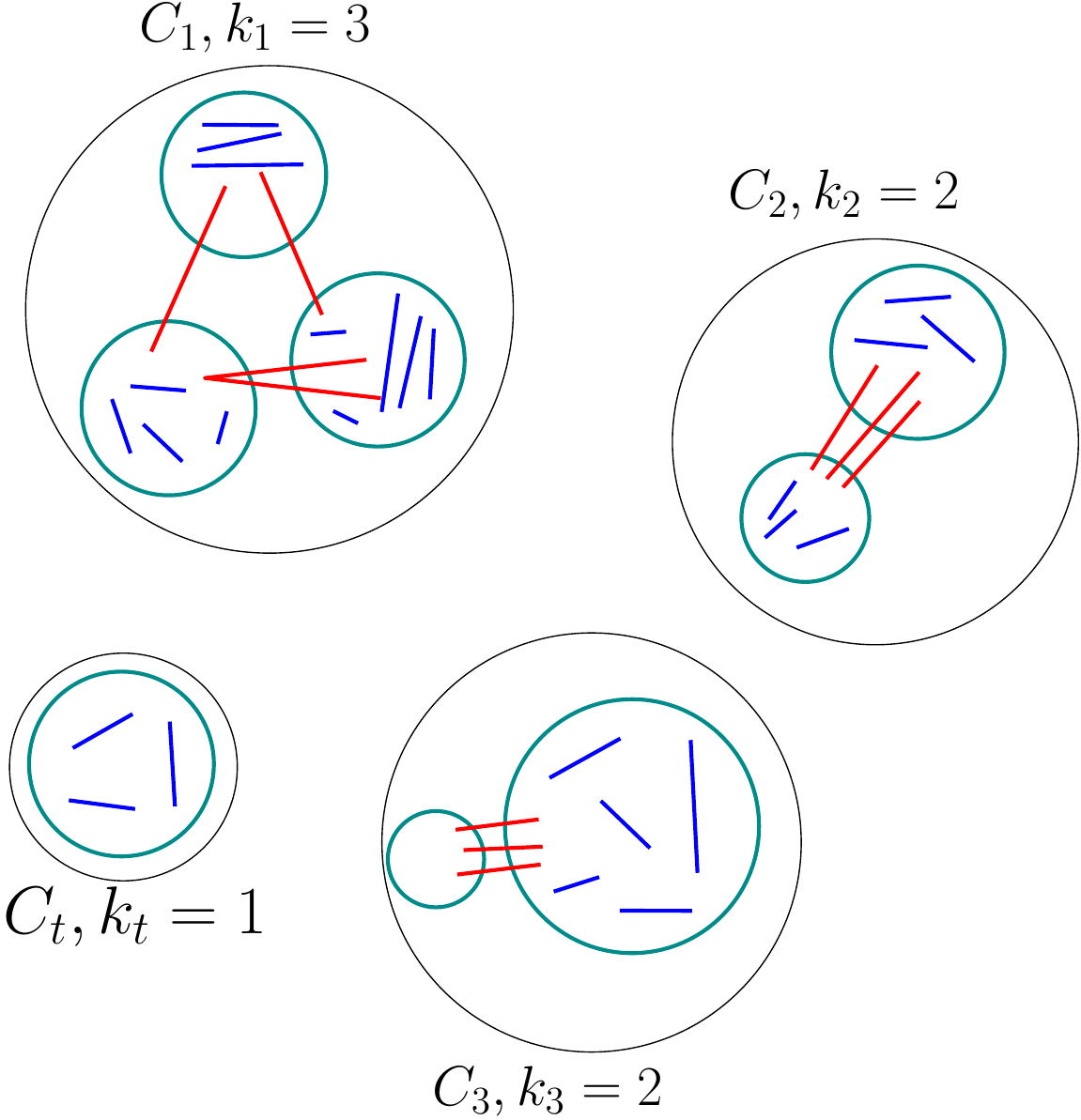}
    \caption{Blue edges are $S_{i}$. Red edges are $\hat{S}_i$. Red and blue edges together are $S_{i-1}$. The red edges in component $C_q$ are $E_q$. Cyan circles inside $C_q$ are components inside $C_q$ after $\hat{S}_i$ is deleted. }
    \label{fig:thorup:decomposition}
\end{figure}
    In the $i$-th iteration, observe that $\hat{S}_i=S_{i-1}\setminus S_i$ is just the edges crossing the partition $P(S_i)$ in the graph
$G'=G(V, S_{i-1})$ remaining from iteration $i-1$. Also,
$\frac{f(S_{i-1})-f(S_{i})}{|S_{i-1}|-|S_i|}=\frac{\kappa(S_i)-\kappa(S_{i-1})}{E(P(S_i))}$. We show that this is the correct value consistent with Thorup's ideal
relative weights that the edges in $\hat{S_i}$ should be set to.

    We proceed inductively on the $i$-th iteration. 

    Throughout the proof see Figure \ref{fig:thorup:decomposition}. Fix $G'=G(V, S_{i-1})$. Let the connected components of $G'$ be $C_1, ..., C_t$ with component $C_q$ having $k_q$ connected components after deleting $\hat{S}_i=S_{i-1}\setminus S_i$ and contributing edges $E_q$ to the cross edges in $\hat{S}_i$. Let $k=\sum_i k_i$. 

    If for some component $C_q$, we have $\frac{E_q}{k_q-1} < \frac{|S_{i-1}|-|S_{i}|}{k-1}$, then for $S_i'=S_{i-1}\setminus E_q$, we have 
    \[
    \frac{|S_{i-1}|-|S_i'|}{\kappa(S_{i-1}-\kappa(S_{i}'))} = \frac{|E_q|}{k_q-1} < \frac{|S_{i-1}|-|S_{i}|}{k-1}
    \]
    A contradiction to the optimality of $S_i$. \\ 

    Hence, for all components $C_q$ with $k_q>1$,  $\frac{E_q}{k_q-1} \geq  \frac{|S_{i-1}|-|S_i|}{k-1}$. But then for $S_i'=S_i\cup E_q$, we have 
    \[
    \frac{|S_{i-1}|-|S_i'|}{\kappa(S_i')-\kappa(S_{i-1})} = \frac{|S_{i-1}|-|S_i|-|E_q|}{\kappa(S_i)-\kappa(S_{i-1}) - k_q + 1} \leq  \frac{|S_{i-1}|-|S_i|}{\kappa(S_i)-\kappa(S_{i-1})}
    \]
    By optimality of $S_i$, equality must hold. 

    Hence it must be that if $k_q>1$ then $(k_1-1)/|E_q|=(k-1)/|\hat{S}_i|$. 
    So the algorithm sets the correct densities for the edges in $\hat{S_i}$ according to Thorup's algorithm. 
\end{proof}


\end{document}